\documentclass[a4paper, UKenglish, cleveref, autoref, thm-restate]{lipics-v2021}

\hideLIPIcs  

\title{A Computer-Assisted Proof of the Optimal Density Bound for Pinwheel Covering}

\author{Akitoshi Kawamura}{Research Institute for Mathematical Sciences, Kyoto University, Japan}{kawamura@kurims.kyoto-u.ac.jp}{https://orcid.org/0009-0006-3706-7470}{Support Center for Advanced Telecommunications Technology Research (SCAT); Kyoto University Foundation}

\author{Yusuke Kobayashi}{Research Institute for Mathematical Sciences, Kyoto University, Japan}{yusuke@kurims.kyoto-u.ac.jp}{https://orcid.org/0000-0001-9478-7307}{JSPS KAKENHI Grant Numbers JP22H05001, 24K02901, 26K21777; JST ERATO Grant Number JPMJER2301}

\authorrunning{A. Kawamura and Y. Kobayashi} 

\Copyright{Akitoshi Kawamura and Yusuke Kobayashi} 



\category{} 

\nolinenumbers 

\usepackage[T1]{fontenc}
\usepackage{amsfonts}
\usepackage{cite}
\usepackage[whole]{bxcjkjatype}

\usepackage{tikz}
\usetikzlibrary{positioning}
\usetikzlibrary{calc}

\newcommand{\Zset}{\mathbb Z}
\newcommand{\density}{\mathrm D}

\begin{document}

\maketitle 

\begin{abstract}
In the covering version of the pinwheel scheduling problem, a daily task must be assigned to agents under the constraint that agent $i$ can perform the task at most once in any $a_i$-day interval. In this paper, we determine the optimal constant $\alpha^* = 1.264\ldots$ such that every instance with $\sum_{i} 1 / a_i \ge \alpha^*$ is schedulable. This resolves an open problem posed by Kawamura and Soejima (2020). 
Our proof combines Kawamura's (2026) techniques for the packing version with new mathematical insights to reduce the analysis to a finite set of instances, which are then verified through an exhaustive computer-aided search that draws on ideas from G\k asieniec, Smith, and Wild (2022). 
The same result was obtained independently by Mishra (2026).
\end{abstract}

\section{Introduction}
\label{sec:intro}

There is a task that must be performed every day, and it needs to be distributed among $k$ agents. 
Each agent $i \in [k] = \{1, \ldots, k\}$ has an associated number $a_i$, 
called its \emph{period}, 
and may be assigned the task at most once in any interval of $a_i$ consecutive days. 
Under this constraint, we want to find a schedule 
that allows the task to be performed indefinitely. 
This problem is known as the \emph{covering version of pinwheel scheduling}, 
or simply \emph{pinwheel covering} \cite{KKK25,KLS26}. 
It was originally introduced as 
\emph{point patrolling} \cite{KS20}, 
but here we prefer the terminology reflecting a contrast with the 
\emph{packing} version of pinwheel scheduling \cite{Hol}, 
in which the constraint is that each agent $i$ must be assigned the task 
\emph{at least} once in any $a_i$-day interval.

Formally, an instance of pinwheel covering
is a nonempty array $(a _i) _{i \in [k]}$ of positive integers,
which we assume to be arranged in non-decreasing order,
and we seek to find a \emph{schedule} $S \colon \Zset \to [k]$ 
(with $S (t)$ specifying the agent that works on day $t$)
satisfying, for all $i \in [k]$, the following 
\emph{frequency condition},  
where $S^{-1}$ denotes the inverse image of $S$:
\begin{quote}
$\bigl| [m, m + a _i) \cap S ^{-1} (i) \bigr| \leq 1$
for each $m \in \Zset$.
\end{quote}
An instance for which a schedule exists is said to be 
\emph{covering-schedulable}
(or simply \emph{schedulable}).
For example, $(2, 2)$, $(2, 4, 8, 8)$, and $(3, 5, 5, 5, 7)$ are
schedulable, 
with a schedule $S$ for the last one given by
\begin{equation*}
S (t) =
\begin{cases}
 1 & \text{for } t \equiv 0, 3, 7, 10, 14, 17,\\
 2 & \text{for } t \equiv 1, 6, 11, 16,\\
 3 & \text{for } t \equiv 2, 8, 13, 18,\\
 4 & \text{for } t \equiv 4, 9, 15, 20,\\
 5 & \text{for } t \equiv 5, 12, 19 \pmod{21}. 
\end{cases}
\end{equation*}

For an instance $A = (a _i) _{i \in [k]}$ to be schedulable,
its \emph{density}
\begin{equation*}
\density (A) = \sum_{i \in [k]} \frac{1}{a _i}
\end{equation*}
must necessarily be at least $1$, 
because $\density (A)$ is the amount of workforce per day 
that the $k$ agents can provide.
This condition is not sufficient: 
$(2, 3, 5)$ is unschedulable, 
though $\frac 1 2 + \frac 1 3 + \frac 1 5 \geq 1$. 
In fact, these three agents cannot even perform the task for just eight consecutive days:
whichever of the eight days the period-$5$ agent covers, 
there are some four consecutive days left, 
which cannot be covered by the other two agents. 
This argument can be applied recursively, 
so that for every $k$, 
the instance $(2 ^{i - 1} + 1) _{i \in [k]}$ is unschedulable, 
even for $2 ^k$ days \cite[Theorem 17]{KS20}. 
The density of this instance approaches
\begin{equation*}
\alpha^* = \sum_{i=1}^{\infty} \frac{1}{2^{i-1}+1} = 1.264 \ldots
\end{equation*}
as $k \to \infty$. 
It has been suspected \cite[Conjecture 18]{KS20} 
(and remained open \cite[Section 5]{Kaw25}) that 
these are the densest unschedulable instances. 
We confirm this conjecture, 
improving on the best known bound of
$1.545\ldots$ \cite[Theorem 16]{KS20}: 

\begin{theorem}\label{thm:main}
Every instance $A$ consisting of positive integers and satisfying $\density (A) \ge \alpha^*$
is covering-schedulable. 
\end{theorem}

We note that Theorem~\ref{thm:main} was also obtained independently by Mishra~\cite{Mis26}.
Compared to~\cite{Mis26}, our proof seems to be simpler: 
apart from the computer-assisted part, the theoretical arguments (Section~\ref{sec:mainthm}) fit within a few pages.

The explicit mention of the integrality of periods in Theorem~\ref{thm:main} is 
because we will later extend the definitions 
and consider real-valued (non-integer) periods. 
This is an adaptation of the idea 
used in the recent proof by Kawamura~\cite{Kaw26} 
of an analogous optimal bound 
for the packing version of the problem, 
which we will review briefly below (Theorem~\ref{thm:mainpack}).

Like this packing version, 
our proof of Theorem~\ref{thm:main} involves 
exhaustive computer search for schedules of finitely many instances. 
Thus, the rest of the paper consists of 
the theoretical part (Section \ref{sec:mainthm}), 
which proves why checking this finite set of instances suffices, 
and the experimental part (Section \ref{sec:computation})
describing the techniques we used to schedule this finite but huge set of instances. 

\subsection*{Related Work: Packing Version}

As we mentioned, a better-studied problem is 
the packing version of pinwheel scheduling~\cite{Hol}, 
where each agent $i$ must be employed \emph{at least} once every $a _i$ days.
Formally, an instance $A=(a_i) _{i \in [k]}$ 
is called \emph{packing-schedulable} if there is 
$S \colon \Zset \to [k]$ that satisfies the following for all $i \in [k]$: 
\begin{quote}
$\bigl| [m, m + a _i) \cap S ^{-1} (i) \bigr| \geq 1$
for each $m \in \Zset$.
\end{quote}
A packing-schedulable instance must have density at most $1$. 
Conversely, there has been extensive research on sufficient conditions on the density
that guarantee packing-schedulability. 
Building on prior work~\cite{CC93, FL02, GSW22, LL97}, Kawamura~\cite{Kaw26} recently 
established the optimal threshold:

\begin{theorem}[Kawamura\mbox{\cite[Theorem 1]{Kaw26}}]\label{thm:mainpack}
Every instance $A$ consisting of positive integers and satisfying $\density (A) \le \frac{5}{6}$ is packing-schedulable. 
\end{theorem}

Our proof of Theorem~\ref{thm:main} will 
partly mirror that of Theorem~\ref{thm:mainpack}, 
augmented by a new idea.
To highlight the contrast, 
we here outline the proof of Theorem~\ref{thm:mainpack} given in~\cite{Kaw26}. 
A key idea is to extend the problem to the setting
where periods are positive real numbers, not just integers.  
An algorithm is then given \cite[Lemma 5]{Kaw26} that, 
for any integer $\theta \ge 1$, 
converts any packing-unschedulable instance $A$ consisting of positive integers 
into a packing-unschedulable instance $B$
consisting of periods in $\{1, 2, \dots , \theta\} \cup (\theta, 2 \theta]$ and satisfying $\density (B) < \density (A) + 1 / (2 \theta)$. 
It thus suffices to prove that for some integer $\theta \geq 1$, 
\begin{quote}
every instance $B = (b _i) _{i \in [k]}$ such that 
$b _i \in \{1, 2, \dots, \theta\} \cup (\theta, 2 \theta]$ for each $i \in [k]$ and 
$\density (B) < \frac{5}{6} + 1 / (2 \theta)$
is packing-schedulable. 
\end{quote}
This claim holds for $\theta = 11$, as 
was shown by an exhaustive computer analysis \cite[Lemma 6]{Kaw26}. 

\section{Proof}
\label{sec:mainthm}

The structure of the proof of Theorem~\ref{thm:main} is shown in Figure~\ref{fig:outline}. 
We first prove its special case 
where we do not have an agent with period~$2$:

\begin{lemma}\label{lem:mainrelaxed}
Every instance $A$ consisting of integers $> 2$ and satisfying $\density (A) \ge \alpha^*$ 
is schedulable. 
\end{lemma}

\begin{figure}[t]
\begin{center}
\begin{tikzpicture}[
    >=latex,
    thm/.style={
        draw,
        rectangle,
        inner xsep=8pt,
        inner ysep=6pt
    },
    sec/.style={
        midway,
        below,
        font=\scriptsize,
        align=center
    }
]

\node[thm] (thm)  at (0,0) {Theorem 1};
\node[thm] (lem3) at (4.4,0) {Lemma 3};
\node[thm] (lem5) at (8.8,0) {Lemma 5};

\draw[->] (lem3) -- node[sec] {Section 2.2} (thm);
\draw[->] (lem5) -- node[sec] {Section 2.1} (lem3);
\draw[->] (12.8,0) -- node[sec] {exhaustive analysis\\(Section 3)} (lem5);

\end{tikzpicture}
      \caption{Outline of the proof of Theorem~\ref{thm:main}.}
  \label{fig:outline}
\end{center}
\end{figure}

\subsection{Instances Without Period 2 (Proof of Lemma~\ref{lem:mainrelaxed})}
\label{subsection: special case}

This part follows the same argument as the proof of Theorem~\ref{thm:mainpack} above. 
We start by extending the problem to the real-valued setting~\cite{KKK25}: 
a nonempty array $(a_i)_{i \in [k]}$ of positive real numbers is \emph{covering-schedulable} (or simply \emph{schedulable}) if 
there exists $S \colon \Zset \to [k]$ such that  
for all $i \in [k]$, 
\begin{quote}
$\left| [m, m+ \lfloor r \cdot a_i \rfloor) \cap S^{-1}(i) \right| \leq r$
for each $m\in \mathbb{Z}$ and $r \in \mathbb{N}$. 
\end{quote}
Note that this condition boils down to the one in Section~\ref{sec:intro} when $a_i$ is an integer. 

The following can be seen as the counterpart of the algorithm used for the packing problem at the end of Section~\ref{sec:intro}. 

\begin{lemma}[Kawamura, Kobayashi, and Kusano\mbox{\cite[Lemma 6]{KKK25}}]\label{lem:fold*}
Let $\theta \ge 1$, and let $A$ be an unschedulable instance. 
We can convert $A$ into another unschedulable instance $B$ 
such that 
\begin{itemize}
    \item any element in $B$ with a value $\leq \theta$ is in $A$ as well,
    \item any element in $B$ is at most $2\theta$, and 
    \item $\density (B) \geq \density (A) - 1/\theta$.
\end{itemize}
\end{lemma}

The first two claims about $B$ in this lemma imply that 
if $A$ consists of integers $> 2$ and $\theta$ is an integer $> 2$, then each element in $B$ belongs to $\{3, 4, \dots , \theta\} \cup (\theta, 2 \theta]$. 
Therefore, just as in the packing version, 
to prove Lemma~\ref{lem:mainrelaxed}, it suffices to prove  
the following for some integer $\theta > 2$: 

\begin{description}
    \item[($\star$)] 
every instance $B = (b _i) _{i \in [k]}$ 
such that $b _i \in \{3, 4, \dots, \theta \} \cup (\theta , 2 \theta]$ for each $i \in [k]$ 
and $\density (B) \ge \alpha^* - 1 / \theta$ 
is schedulable.     
\end{description}

As it turns out, this claim holds for $\theta = 10$. 
In fact, 
the instance $C = (\lceil b_i \rceil) _{i \in [k]}$
obtained by rounding up each period in $B$ is still schedulable, 
as stated in the following lemma.
Note that although $C$ usually has density a bit lower than $B$, 
it satisfies $\density'(C) \ge \density (B) \ge \alpha^* - \frac{1}{10}$, where
\begin{equation*}
\density' \bigl( (c _i) _{i \in [k]} \bigr) = \sum_{i \in [k]} 
\begin{cases}
1 / c _i  & \text{for $c_i \le 10$}, \\ 
1 / (c _i - 1)  & \text{for $c_i > 10$}. 
\end{cases}    
\end{equation*}

\begin{lemma}\label{lem:04}
Every instance $C$ consisting of periods in $\{3, 4, \dots , 20\}$ and satisfying
$\density'(C) \ge \alpha^* - \frac{1}{10}$
is schedulable.
\end{lemma}

Note that this lemma is essentially about 
finitely many instances,
because instances with $\geq 20$ agents with periods $\leq 20$ are 
trivially schedulable. 
It can thus be verified exhaustively by computer; see Section~\ref{sec:computation}. 
We remark that 
replacing $10$ in Lemma~\ref{lem:04}
(and in the definition of $\density'$)
by a smaller number, say $9$
(and accordingly replacing $20$ by $18$), 
would make it false.
For example,
the instance $(3, 4, 10, 10, 10, 12, 13, 17)$ is 
unschedulable
(as can be checked by brute force computer search),
even though $\frac 1 3 + \frac 1 4 + \frac{1}{9} + \frac{1}{9} + \frac{1}{9} + \frac{1}{11} + \frac{1}{12} + \frac{1}{16} = \frac{203}{176} > \alpha ^* - \frac 1 9$. 

We have thus proved Lemma~\ref{lem:mainrelaxed}. 
Note that the same approach alone cannot prove Theorem~\ref{thm:main} directly: 
no integer $\theta$ satisfies the claim ($\star$)
if $\{3, 4, \dots, \theta\}$ is replaced by $\{2, 3, \dots, \theta\}$, 
as can be seen by the unschedulable (see Section~\ref{sec:intro}) instance 
$(2 ^{i - 1} + 1) _{i \in [k]}$ 
for the largest $k$ with $2^{k-1} + 1 \le 2\theta$.

\subsection{Eliminating Period 2 (Proof of Theorem~\ref{thm:main} Assuming Lemma~\ref{lem:mainrelaxed})}
\label{subsection: eliminating period 2}

Now that we have proved Lemma~\ref{lem:mainrelaxed}, 
we are left with (integral) instances 
that have an agent with period~$2$. 
Luckily, 
such an instance readily reduces to a smaller one: 

\begin{lemma}
\label{lemma: inserting 2}
An instance $(a _i) _{i \in [k]}$ with $a _1 = 2$ is schedulable
if $(\lceil a _{i + 1} / 2 \rceil) _{i \in [k - 1]}$ is. 
\end{lemma}

\begin{proof}
Given a schedule $S' \colon \Zset \to [k - 1]$ for the latter, 
we can obtain a schedule $S \colon \Zset \to [k]$ for the former  
by inserting the period-$2$ agent every second day, 
i.e., by defining $S (t) = 1$ for odd $t$ and $S (t) = S' (t / 2) + 1$ for even $t$. 
\end{proof}

To prove Theorem~\ref{thm:main}, 
we apply this reduction repeatedly,
until we obtain an instance not starting with period $2$---that is, 
either an instance containing period $1$, which is trivially schedulable, 
or one consisting of periods $> 2$, 
whose schedulability is hopefully guaranteed by Lemma~\ref{lem:mainrelaxed}. 
The problem is that this last instance may have density $< \alpha ^*$, 
as the reduction step may cause the density to drop from, say, 
$\density (2, 3, 5, 5, 21) = 1.28\ldots \geq \alpha ^*$ 
to $\density (2, 3, 3, 11) = 1.25\ldots < \alpha ^*$. 
We will argue that this can happen only when we are on the easier track ending up with period~$1$. 

\begin{proof}[Proof of Theorem~\ref{thm:main}]
The \emph{type} of an instance $A = (a _i) _{i \in [k]}$
is defined as the largest (nonnegative) integer $p \leq k$ such that 
$a _i \leq 2 ^i$ for all $i \in [p]$.
This instance $A$ is said to be \emph{head-dense} if
$a _i \leq 2 ^{i - 1}$ for some $i \in [p]$, 
and \emph{tail-dense} if 
\begin{equation*}
\sum _{i = p + 1} ^k \frac{1}{a _i} \geq \sum _{i = p + 1} ^\infty \frac{1}{2 ^{i - 1} + 1}. 
\end{equation*}
We prove by induction on $p$ that
\emph{every integral instance $A = (a _i) _{i \in [k]}$ of type $p$ 
that is head-dense or tail-dense is schedulable}.
This implies the theorem, 
because an instance that is neither head-dense nor tail-dense has 
density $< \alpha ^*$. 

The claim for $p = 0$ is already proved, 
since 
(there is no head-dense instance of type~$0$, and)
tail-dense instances of type~$0$ were shown schedulable in Lemma~\ref{lem:mainrelaxed}. 
So suppose $p > 0$.
If $a _1 = 1$, then $A$ is trivially schedulable, 
so suppose $a _1 = 2$. 
We claim that the instance $B = (\lceil a _{i + 1} / 2 \rceil) _{i \in [k - 1]}$, 
which has type $p - 1$, 
is also head-dense or tail-dense. 
Once we have proved this, 
$B$ is schedulable by the induction hypothesis, 
and so is $A$ by Lemma~\ref{lemma: inserting 2}.

If $A$ is head-dense, so is $B$ and we are done. 
Suppose otherwise, so that $A$ is tail-dense. 
For each $i \geq p$, 
the $i$th period $\lceil a _{i + 1} / 2 \rceil$ in $B$ 
is roughly half of the corresponding period $a _{i + 1}$ in $A$: specifically, 
since $a _{i + 1} \ge a _{p + 1} > 2 ^{p + 1}$, 
we have $\lceil a _{i + 1} / 2 \rceil \leq a _{i + 1} \cdot (2 ^p + 1) / (2 ^{p + 1} + 1)$.
This implies the desired tail-density of $B$ by 
\begin{equation*}
\label{equation: proving tail density by induction}
 \sum _{i = p} ^{k - 1} \frac{1}{\lceil a _{i + 1} / 2 \rceil} 
\geq 
 \frac{2 ^{p + 1} + 1}{2 ^p + 1} \cdot \sum _{i = p + 1} ^k \frac{1}{a _i}
\geq 
 \frac{2 ^{p + 1} + 1}{2 ^p + 1} \cdot \sum _{i = p + 1} ^\infty \frac{1}{2 ^{i - 1} + 1}
\geq 
 \sum _{i = p} ^\infty \frac{1}{2 ^{i - 1} + 1}, 
\end{equation*}
where the last inequality follows from Lemma~\ref{lem:densityinequality} below.
\end{proof}

\begin{lemma}\label{lem:densityinequality}
    For a positive integer $p$, we have
\begin{equation*}
 \frac{2 ^p}{2 ^p + 1} \cdot \sum _{i = p + 1} ^\infty \frac{1}{2 ^{i - 1} + 1}
\geq 
 \frac{1}{2 ^{p - 1} + 1}.
\end{equation*}
\end{lemma}

\begin{proof}
For $i = 1$, $2$, \ldots, let 
\begin{equation*}
 \gamma _i 
=
  \frac{2 ^i}{2 ^i + 1}
 \cdot
  \frac{2 ^{i - 1} + 1}{2 ^{i + 1} + 1}. 
\end{equation*}
We have $\gamma _1 \geq \gamma _2 \geq \cdots$, because $\gamma _i \geq \gamma _{i + 1}$ follows from
\begin{multline*}
  (\gamma _i - \gamma _{i + 1}) \cdot (2 ^i + 1) \cdot (2 ^{i + 1} + 1) \cdot (2 ^{i + 2} + 1) \\
 = 
  2 ^i \cdot (2 ^{i - 1} + 1) \cdot (2 ^{i + 2} + 1) - 2 ^{i + 1} \cdot (2 ^i + 1) ^2
 =
  2 ^i \cdot (2 ^{i - 1} - 1)
 \geq 
  0.
\end{multline*}
In particular, 
for each $i > p$, 
we have $\gamma_p \geq \gamma_i$, and hence
\begin{equation*}
  \sum _{i = p + 1} ^\infty \frac{\gamma _p}{2 ^{i - 1} + 1}
 \geq 
  \sum _{i = p + 1} ^\infty 
   \frac{\gamma _i}{2 ^{i - 1} + 1}
 = 
  \sum _{i = p + 1} ^\infty 
   \biggl( \frac{1}{2 ^i + 1} - \frac{1}{2 ^{i + 1} + 1} \biggr)
 = 
  \frac{1}{2 ^{p + 1} + 1}.
\end{equation*}
Multiplying by $(2 ^{p + 1} + 1) / (2 ^{p - 1} + 1)$, 
we obtain the claimed bound.
\end{proof}

\section{Computational Techniques (Verifying Lemma~\ref{lem:04})}
\label{sec:computation}

In scheduling an instance $A = (a _i) _{i \in [k]}$,
all that matters on a given day is 
the number of days each agent~$i$ has to wait before it can work again, 
which is a nonnegative integer $< a _i$. 
Formally, 
a \emph{state} is a $k$-tuple $(u _i) _{i \in [k]}$
with $u _i \in \{0, \ldots, a _i - 1\}$, 
and there is a transition from state $(u _i) _{i \in [k]}$ to state $(v _i) _{i \in [k]}$ with label $j \in [k]$ 
when $u _j = 0$ and
\begin{equation*}
    v _i = 
    \begin{cases}
        a _i - 1 & \text{if} \ i = j,
    \\
        \max \{0, u _i - 1\} & \text{otherwise}.
    \end{cases}
\end{equation*} 
A schedule is (the sequence of labels of) an infinite walk in this state transition graph
(which has $a _1 \dotsm a _k$ states).
Hence, $A$ is schedulable if and only if this graph has a cycle. 
This allows us to decide schedulability of a given instance 
in a finite, albeit exponential, amount of time. 

We can thus in principle verify Lemma~\ref{lem:04} exhaustively. 
Yet, 
doing so in a reasonable amount of time 
calls for nontrivial techniques, 
some of which we describe briefly below. 
A computer program implementing these ideas
is publicly available on GitHub at
\url{https://github.com/a7kawamura/pinwheel_solver}.

\subparagraph*{Representation of States}
Since we are interested in instances with periods $\leq 20$, 
and the hard ones are those with many (i.e., $13$--$19$) agents, 
they typically have multiple agents with the same period. 
We may exploit this fact (just as for the packing version \cite[Section 4.2.2]{GSW22}) by 
identifying states up to permutation of entries
corresponding to duplicate periods. 
This significantly reduces the size of the state transition graph, 
and also often the length of the repeating pattern of the schedule to be found. 
For example, 
the $21$-day cycle of the schedule $S$ in Section~\ref{sec:intro}
for the instance $(3, 5, 5, 5, 7)$ 
can be regarded as repeating the $7$-day cycle of 
employing agents with periods $3$, $5$, $5$, $3$, $5$, $7$, $5$ in order, 
with the understanding that the three period-$5$ agents 
always work in a round-robin fashion. 

\subparagraph*{Reducing the Number of Instances}
We also make some effort 
to get away with fewer instances to check. 
We need not check an instance that has an agent 
without whom the instance still satisfies the modified density bound in Lemma~\ref{lem:04}.
We can also ignore an instance containing period $10$, 
because replacing $10$ by $11$ 
does not affect the modified density. 
That leaves $25\,242\,331$ instances,
still too many to run the exponential-time cycle-detecting algorithm on.

To further reduce the number of instances, 
we use the fact that
a schedule for an instance $C = (c _i) _{i \in [k]}$ can be easily constructed 
from one for the instance $C'$ obtained by 
replacing the last two agents with periods $c _{k - 1}$ and $c _k$ 
by a single agent with period $\min \{c _{k - 1}, \lceil c _k / 2 \rceil\}$.
Since $C'$ has one agent fewer than $C$, 
it is often computationally easier to schedule, 
although of course there is a risk that $C'$ may be unschedulable while $C$ is schedulable.
Thus, in order to prove $C$ schedulable, 
we run the cycle-detecting algorithm in parallel
on instances $C$, $C'$, $C''$, \ldots 
(some of which may not satisfy the bound in Lemma~\ref{lem:04}), 
until one of them turns out schedulable. 
A similar idea was
already used for the packing version \cite[Section 5.2.2]{GSW22}. 

This operation of ``folding'' causes many instances 
to be proved schedulable via a common short instance, 
significantly reducing the number of instances on which 
we actually need to run the cycle-detecting algorithm. 
We ended up running it
only $11\,000$--$12\,000$ times
(this number fluctuates 
depending on which of the parallel searches for schedules finishes first).

\bibliography{pinwheel}

@article{CC93, 
  author = {Chan, M. Y. and Chin, Francis},
  title = {Schedulers for Larger Classes of Pinwheel Instances}, 
  journal = {Algorithmica},
  volume = {9}, 
  pages = {425--462},
  year = {1993}, 
  doi = {10.1007/BF01187034},
}

@inproceedings{GSW22,
author = {Leszek Gąsieniec and Benjamin Smith and Sebastian Wild},
title = {Towards the 5/6-Density Conjecture of Pinwheel Scheduling},
booktitle = {2022 Proceedings of the Symposium on Algorithm Engineering and Experiments (ALENEX)},
year = {2022},
chapter = {},
pages = {91-103},
doi = {10.1137/1.9781611977042.8},
URL = {https://epubs.siam.org/doi/abs/10.1137/1.9781611977042.8}
}

@INPROCEEDINGS{Hol,
  author={Holte, R. and Mok, A. and Rosier, L. and Tulchinsky, I. and Varvel, D.},
  booktitle={Proceedings of the Twenty-Second Annual Hawaii International Conference on System Sciences. Volume II: Software Track}, 
  title={The pinwheel: a real-time scheduling problem}, 
  year={1989},
  volume={2},
  number={},
  pages={693-702},
  doi={10.1109/HICSS.1989.48075}
}

@article{Kaw26,
    author = {Kawamura, Akitoshi},
    title = {Proof of the Density Threshold Conjecture for Pinwheel Scheduling},
    journal = {Proceedings of the National Academy of Sciences},
    year = {2026},
    volume = {123}, 
    number = {32},
    pages = {e2530214123},
    doi = {10.1073/pnas.2530214123},
    note = {Preliminary version in the \emph{Proceedings of the 56th Annual ACM Symposium on Theory of Computing (STOC)}, pages 1816--1819.}
}

@InCollection{Kaw25,
    author = {Akitoshi Kawamura},
    title = {Perpetual Scheduling under Frequency Constraints},
    booktitle = {Algorithmic Foundations for Social Advancement: Recent Progress on Theory and Practice},
    editor = "{Shin-ichi} Minato and Takeaki Uno and Norihito Yasuda and Takashi Horiyama and {Ken-ichi} Kawarabayashi and Shigeru Yamashita and Hirotaka Ono", 
    year = 2025,
    doi = {10.1007/978-981-96-0668-9_23},
    publisher = {Springer},
}

@article{KS20,
author="Kawamura, Akitoshi and Soejima, Makoto",
title = {Simple Strategies Versus Optimal Schedules in Multi-Agent Patrolling},
journal = {Theoretical Computer Science},
volume = {839},
pages = {195-206},
year = {2020},
doi = {10.1016/j.tcs.2020.07.037},
}

@inproceedings{KKK25,
  author       = {Akitoshi Kawamura and
                  Yusuke Kobayashi and
                  Yosuke Kusano},
  editor       = {Irene Finocchi and
                  Loukas Georgiadis},
  title        = {Pinwheel Covering},
  booktitle    = {Algorithms and Complexity -- 14th International Conference, {CIAC}
                  2025, Rome, Italy, June 10--12, 2025, Proceedings, Part {II}},
  series       = {Lecture Notes in Computer Science},
  volume       = {15680},
  pages        = {185--199},
  publisher    = {Springer},
  year         = {2025},
  url          = {https://doi.org/10.1007/978-3-031-92935-9\_12},
  doi          = {10.1007/978-3-031-92935-9\_12},
  bibsource    = {dblp computer science bibliography, https://dblp.org}
}

@article{FL02,
  author       = {Peter C. Fishburn and
                  J. C. Lagarias},
  title        = {Pinwheel Scheduling: Achievable Densities},
  journal      = {Algorithmica},
  volume       = {34},
  number       = {1},
  pages        = {14--38},
  year         = {2002},
  doi          = {10.1007/S00453-002-0938-9},
  bibsource    = {dblp computer science bibliography, https://dblp.org}
}

@article{LL97,
  title={A Pinwheel Scheduler for Three Distinct Numbers with a Tight Schedulability Bound},
  author={Shun-Shii Lin and Kwei-Jay Lin},
  journal={Algorithmica},
  year={1997},
  volume={19},
  pages={411-426},
  doi={10.1007/PL00009181}
}

@article{KLS26,
title = {Hardness and fixed parameter tractability for pinwheel scheduling problems},
journal = {Theoretical Computer Science},
volume = {1077},
pages = {115998},
year = {2026},
issn = {0304-3975},
doi = {https://doi.org/10.1016/j.tcs.2026.115998},
author = {Yusuke Kobayashi and Bingkai Lin and Joseph Swernofsky},
}

@inproceedings{Mis26,
  author       = {Ahan Mishra},
  title        = {An Optimal Density Bound for Discretized Point Patrolling},
  booktitle    = {Proceedings of the ACM-SIAM Symposium on Discrete Algorithms (SODA)},
  year         = {2026}, 
  pages        = {2846--2875},
  doi = {10.1137/1.9781611978971.105}
}

\end{document}